\documentclass[11pt]{article}

\usepackage{relsize}
\usepackage[blocks]{authblk}
\usepackage{enumitem}
\setlist{nosep}

\newif\ifllncs
\llncsfalse

\newcommand{\authorrunning}[1]{} % {\fancyhead[L]{#1}}
\newcommand{\titlerunning}[1]{} % {\fancyhead[R]{#1}}
\newcommand{\institute}[1]{\affil{#1}}
\newcommand{\inst}[1]{}
\newcommand{\email}[1]{\texttt{#1}}
\newcommand{\orcidID}[1]{$^{\textrm{[#1]}}$}
\newcommand{\keywords}[1]{}
\date{}

\newtheorem{lemma}{Lemma}
\newtheorem{proposition}{Proposition}
\newtheorem{conjecture}{Conjecture}

\newenvironment{proof}[1][\unskip]{\paragraph{Proof \textit{#1}.}}{}
\newcommand{\qed}{\hfill\ensuremath{\scriptstyle \blacksquare}}

\usepackage[letterpaper,margin=1in]{geometry}
\usepackage{graphicx}
\usepackage{xspace}
\usepackage{amsmath}
\usepackage{multirow}
\usepackage{amssymb}
\usepackage{thm-restate}
\usepackage{booktabs}
\usepackage{siunitx}
\usepackage{comment}

\usepackage{hyperref}
\usepackage{xcolor}
\hypersetup{
    colorlinks,
    linkcolor={red!50!black},
    citecolor={blue!70!black},
    urlcolor={blue!20!black}
}

\newcommand{\kmer}{$k$-mer\xspace}
\newcommand{\kmers}{\kmer{}s\xspace}
\newcommand{\fs}{\ensuremath{\varphi}\xspace} % Name of the selection function (''fonction de selection'' \fs, as \sf is already taken)
\DeclareMathOperator{\sketch}{\mathcal{M}}
\DeclareMathOperator{\powerset}{\mathcal{P}}

\DeclareMathOperator{\lc}{lc} % Left companion set
\DeclareMathOperator{\rc}{rc} % Right companion set
\DeclareMathOperator{\hit}{\mathcal{H}} % Hitting number
\DeclareMathOperator{\sig}{\mathcal{S}} % Cycle signature
\DeclareMathOperator{\sigim}{\mathcal{I}}% I-move signature
\newcommand{\gmds}{\ensuremath{G_{\textrm{MDS}}}\xspace} % MDS graph
\newcommand{\gcomp}{\ensuremath{G_{\text{comp}}}\xspace} % Component graph
\newcommand{\gpcr}{\ensuremath{G_{\text{PCR}}}\xspace} % Graph of non-decycling PCR sets
\newcommand{\restrict}[1]{\raisebox{-.5ex}{\ensuremath|}_{#1}}
\newcommand{\im}[2]{#1\restrict{#2}} % I-move formatting
\newcommand {\dprime}{''} % Stupid but '' messes up with highlighting (looks like the end of a string)
\newcommand{\SymDiff}{\mathbin{\Delta}} % Set symmetric difference

\begin{document}

\title{Sketching methods with small window guarantee using minimum decycling sets}

\author{Guillaume~Mar\c{c}ais\inst{1}\orcidID{0000-0002-5083-5925} \and
Dan~DeBlasio\inst{1}\orcidID{0000-0003-4110-4431} \and
Carl~Kingsford\inst{1}\orcidID{0000-0002-0118-5516}}
\authorrunning{G.\ Mar\c{c}ais et al.}
\institute{Computational Biology Department, Carnegie Mellon University, \\
  Pittsburgh PA 15213, USA \\
  \email{\{gmarcais,deblasio,carlk\}@cs.cmu.edu}}

\maketitle              % typeset the header of the contribution

\begin{abstract}
  Most sequence sketching methods work by selecting specific \kmers from sequences so that the similarity between two sequences can be estimated using only the sketches.
  Because estimating sequence similarity is much faster using sketches than using sequence alignment, sketching methods are used to reduce the computational requirements of computational biology software packages.
  Applications using sketches often rely on properties of the \kmer selection procedure to ensure that using a sketch does not degrade the quality of the results compared with using sequence alignment. Two important examples of such properties are locality and window guarantees, the latter of which ensures that no long region of the sequence goes unrepresented in the sketch.\smallskip

  A sketching method with a window guarantee, implicitly or explicitly, corresponds to a \emph{Decycling Set}, an unavoidable sets of \kmers.
  Any long enough sequence, by definition, must contain a \kmer from any decycling set (hence, it is unavoidable).
  Conversely, a decycling set also defines a sketching method by choosing the \kmers from the set as representatives.
  Although current methods use one of a small number of sketching method families, the space of decycling sets is much larger, and largely unexplored.
  Finding decycling sets with desirable characteristics (e.g.\@, small remaining path length) is a promising approach to discovering new sketching methods with improved performance (e.g.\@, with small window guarantee).\smallskip

  The \emph{Minimum Decycling Sets} (MDSs) are of particular interest because of their minimum size.
  Only two algorithms, by Mykkeltveit and Champarnaud, are previously known to generate two particular MDSs, although there are typically a vast number of alternative MDSs.
  We provide a simple method to enumerate  MDSs.
  This method allows one to explore the space of MDSs and to find MDSs optimized for  desirable properties.
  We give evidence that the Mykkeltveit sets are close to optimal regarding one particular property, the remaining path length. A number of conjectures and computational and theoretical evidence to support them are presented.\smallskip

  Code available at \url{https://github.com/Kingsford-Group/mdsscope}.

\end{abstract}

\section{Introduction}

Sketching methods, such as minimizers~\cite{minimizers1} or open-syncmers~\cite{syncmers}, distill a long sequence into a smaller ``sketch,'' a set of \kmers and their positions in the sequence.
By comparing these sketches, it is possible to quickly estimate whether two sequences are similar and may have a good quality alignment between them, or not.
Because sketching methods greatly reduce the computational needs in many genomics algorithms with usually little impact on the quality of the result, they are used in many computational biology software packages (see~\cite{minimizers-review} for a review).

For our purposes, a \kmer sketching method is modeled by a function \fs that takes a \emph{context} as an input (a substring of the input sequence of fixed length $c$) and outputs a set of positions within the context of the selected \kmers.
The output of \fs can be the empty set, meaning that nothing is selected in this context.
The sketch $\sketch_{\fs}(S)$ for a sequence $S$ is the union of all selected positions over all the contexts of $S$ (see Section~\ref{sec:notations}).
This sketch contains a subset of all the \kmers in $S$ as the function \fs might not pick any \kmer in a context or adjacent contexts may pick the same locations.

The two properties of sketching methods that downstream applications rely on to prove correctness are:
\begin{description}
  \item[Locality]
        The property that similar sequences (i.e., that have reasonably long identical subsequences) will have common elements in their sketches, and hence long enough matches will be detected using the sketches.
        This is naturally satisfied because the selection is done using a function (\fs), therefore two sequences that share an exact substring of length at least $c$ will select the same \kmers in that context.
  \item[Window guarantee]
        The maximum distance $w$ between two selected \kmers is the \emph{window} size or guarantee.
        A small window size guarantees that no large part of a sequence is ignored.
        Equivalently, the window property means \kmers are selected at approximately regular intervals.
\end{description}

Sketching methods are usually optimized for two metrics, \emph{density}~\cite{winnowing} and \emph{conservation}~\cite{syncmers}.
The density is the relative size of the sketch, formally defined as $|\sketch_{\fs}(M)|/|S|$.
A lower density is desirable as a smaller sketch usually implies less computation and lower memory requirements.
The conservation is the proportion of elements that are common between a sketch of $S$ and a sketch of a slightly mutated sequence $S'$, where the common elements are either \kmers or subsequences covered by these \kmers.
Higher conservation is desirable because it usually correlates to higher sensitivity to detect sequence similarities in the face of mutations and errors.
For a fixed $k$, a smaller context size leads to higher conservation, as the presence of a \kmer in the sketch of the mutated $S'$ may be affected by mutations in the entire context~\cite{TheoryLocalKmer2022}.
% G: Yes, we may have to say something here. But'' `for a fixed $k$`, it is not true that smaller context imply larger sketch (equivalently higher density).
% but will also typically lead to a larger sketch. %DD 20231026

% TODO: ref TheoryLocalKmer2022 does not seem to be present.

Not all sketching methods satisfy the window guarantee property (i.e.\@, for some sketching methods, there are infinitely long sequences $S$ with an empty sketch; see Section~\ref{sec:wind-guar-exist}).
However, sketching methods that do not satisfy the window property are problematic in two ways.
First, most algorithms using a sketching method do not have a proof of correctness in cases without the window property (e.g.\@, an aligner may miss arbitrarily long, good quality alignments, preventing claims of sensitivity).

Second, the sketch optimization problem is ill-formed without the window property.
The empty selection function that returns the empty set for any input sequence satisfies vacuously the locality property, it has perfect conservation, and it has the lowest possible density. But of course, 
no information is preserved in an empty sketch and this trivial solution is not useful.
The existence of trivial solutions is not a purely theoretical concern.
When optimizing sketching methods using machine learning, almost empty (and not practically useful) solutions are found if no window constraint is used in the loss function~\cite{maskedminimizers}.

A set of \kmers $M$ is \emph{unavoidable} if any infinitely long sequence must have \kmers from $M$.
Because any sequence uniquely corresponds to a path in the de~Bruijn graph $D_{k}$ of order $k$, an equivalent point of view is the \emph{decycling} sets (DS): $M$ is an unavoidable set of \kmers (and a decycling set) if and only if $D_{k} \setminus M$, the de~Bruijn graph $D_{k}$ with the \kmers from $M$ removed, is a  directed acyclic graph (DAG).

There is a strong two way connection between such decycling sets and sketching methods with a window guarantee.
Consider the set $M_{\fs}$ of possibly selected \kmers (the union of all \kmers selected over every possible context) for sketching method \fs.
 If the sketching method has a window guarantee, then $M_{\fs}$ is a decycling set.
Moreover, the window size of \fs is equal to the \emph{remaining path length} of $M_{\fs}$, i.e.\@, the length of the longest path in the DAG $D_{k} \setminus M_{\fs}$.

The function \fs of a sketching method with the smallest possible context ($c = k$, aka \emph{context-free} methods, such as syncmers) is equivalent to the indicator function of its set $M_{\fs}$: as the input context contains only one \kmer, the output of \fs is not empty exactly when the input \kmer is in $M_{\fs}$.
A sketching method with a larger context may not select every occurrence of \kmers in $M_{\fs}$ from $S$.
For example, a context may contain multiple \kmers from $M_{\fs}$ but the function \fs only selects one of them~\cite{deblasio_practical_2019}.
In other words, given two sketching methods, one context-free and one with a context, having the same set of possibly selected \kmers, the method with a context can lower its density at the expanse of having a lower conservation.
Conversely, given a decycling set $M$, the indicator function of $M$ defines a context-free sketching method with a window guarantee.

This connection between decycling sets and sketching methods suggests, first, that the properties of the decycling sets ultimately define the properties of the associated sketching method.
In other words, by studying the space of decycling sets we gain insights into the design space of sketching methods.
Second, the space of decycling sets is much larger than the decycling sets generated by the few families of sketching methods currently used.
Rather than creating \emph{ad hoc} sketching methods, a promising strategy is to  find a decycling set with desirable properties and use the  sketching method associated with this set.

In this study we focus on minimum-size decycling sets (MDS).
MDSs provide a logical starting point for the study of decycling sets.
First, the MDSs are by definition as small as possible, therefore reducing as much as possible the cost of storing and querying such a set.
Second, these sets are likely to have short remaining path lengths, corresponding to sketching methods with small window guarantee.

After describing the window guarantee of common sketching methods, we describe the structure of the de~Bruijn graph and of its cycles.
We then give two simple graph operations that can be used to enumerate MDSs.
Provided Conjecture~\ref{conj:imove-conn} is true (for which we provide ample theoretical and experimental evidence), all MDSs can be reached with these operations.
Using these operations we design an optimization procedure to find MDSs with short remaining path lengths.
This optimization procedure gives further insight on the range of possible window guarantee for sketching methods and on the of the well-known Mykkeltveit set.

The conjectures and optimization methods proposed here are the basis to further the understanding of MDSs and the design space of the sketching methods that are central to computational biology algorithms, in particular sketching methods with a small context and a strong window guarantee.

\section{Preliminaries and notations}\label{sec:notations}

% XXX: The weird thing about these definitions, is many of them are not used in
% the text, but the concepts was mentioned in the intro.

An alphabet is a small set $\Sigma$ of size $\sigma = |\Sigma|$.
Although the results generalize to any alphabet size, we consider the binary alphabet $\Sigma = \{0, 1\}$ and the DNA alphabet $\{A, C, G, T\}$ of size $4$.
A sequence $S$ is an element of $\Sigma^{*}$, and sequences are indexed starting at $1$.
$S[a:k]$ represent the subsequence starting at position $a$ of length $k$, i.e.\@, the $a$th \kmer of $S$.
$[n]$ is the set of integers $\{1, \ldots, n\}$.

A sketching scheme is defined by its selection function $\fs: \Sigma^{c} \rightarrow \powerset([c - k + 1])$, where $\mathcal{P}$ denotes the power set.
The contexts of $S$ are all the subsequences of length $c$: $S[c] = \{ S[i:c] \mid i \in [|S|-c+1] \}$.
The sketch of $S$ is the set of the positions of the selected \kmers in $S$: $\sketch_{\fs}(S) = \bigcup_{s\in S[c]} \{ i + o \mid o \in \fs(s) \}$.
The set of all possibly selected \kmers for the sketching method $\fs$ is $M_{\fs} = \bigcup_{{s\in S[c]}} \{ s[o:k] \mid o \in \fs(s) \}$.

The de~Bruijn graph of order $k$ is the directed graph $D_{k}= (\Sigma^{k}, E_{k})$, where each \kmer is a node and the edges $u \rightarrow v$ represent the suffix-prefix relationship $u[2:k-1] = v[1:k-1]$.
The de~Bruijn graph is $\sigma$-regular, Eulerian and Hamiltonian.
For convenience, short strings, such as \kmers, are commonly represented as based-$\sigma$ numbers.

\section{Window guarantee of existing sketching schemes}\label{sec:wind-guar-exist}

We review sketching methods commonly used in computational biology and evaluate their window guarantee.

\paragraph{Hash-based methods.}

Hash methods use a hash function $h$ and select the \kmers $m$ that satisfy, for example, $h(m) = 0 \mod p$ or $h(m) < t$ for some predefined constants $p,t$~\cite{KarpRabin,UniverseMinimizers}.
Effectively the hash function randomizes the \kmers and the criteria selects a subset of the \kmers.
Other methods apply a sketching method like minimizers or syncmers and further down-sample the sketch using a hash function~\cite{FractionalHittingSets,syncmers}.

In general these methods do not have a window guarantee and, historically, this was one of the motivations for Schleimer~\cite{winnowing} to introduce the \emph{winnowing scheme} (which is equivalent to minimizers).
Although these schemes can have low density and have a short context ($c = k$), it is achieved at the cost of having no window guarantee.
For example, by choosing low values of the threshold $t$, the density can be made arbitrarily low, but the number of distinct cyclic sequences not covered by the scheme increases dramatically.

\paragraph{Window-based methods.}

These methods always pick at least one \kmer in each context, therefore the context and the window guarantee are closely linked.

The minimizer scheme has three parameters $(k, w, \mathcal{O})$ and in each window of $w$ consecutive \kmers (i.e., the context is a substring of length $w+k-1$), the selection function returns the position of the smallest \kmer according to the order $\mathcal{O}$~\cite{minimizers1,minimizers2}.
There are many ways to select the order $\mathcal{O}$~\cite{polarsets,miniception,DeepMinimizer,winnowmap}, for example to improve the density, but because the selection function never returns the empty set, all these methods have a window guarantee of $w$, independent of the choice of $\mathcal{O}$.

The density of minimizers schemes is usually between $1.5/(w+1)$ and $2/(w+1)$~\cite{improvingminimizers,asymptoticminimizers}, and the context length is $c = w+k-1$.
Density can be lowered by increasing $w$, although this increases the context length (hence weakens the locality and lowers the conservation).
Having a coupling between the window guarantee and the context length constrains the parameter choices for minimizer schemes.

Compared to minimizers, the minmers scheme~\cite{minmers} adds a fourth parameter $d$: in each window of $dw$ consecutive \kmers the selection function returns the position of the $d$ smallest \kmers according to $\mathcal{O}$.
Minmers achieve a density closer to $1/w$ while having a significantly longer context of $dw+k-1$.

\paragraph{Positional minimums.}

Under this generic name are methods such as open-syncmers~\cite{syncmers}, masked minimizers~\cite{maskedminimizers} and parameterized syncmers~\cite{ParameterizedSyncmer}.
These schemes have four parameters $(k, s, \mathcal{O}, m)$ where $s \le k$ and $m$ is a non-empty bit-mask of length $k$.
A context of length $c=k$ is selected if the smallest $s$-mer in the context (choose left-most to break ties) is at position $i$ and bit $i$ is set in the mask $m$.

Whether these schemes have a window guarantee depends on whether the first bit of $m$ is set.
If the first bit is set and a \kmer is selected, then this implies that an $s$-mer at position $i>1$ is strictly smaller than the $s$-mer at position $1$, forming a decreasing list of $s$-mers.
As the \kmers are shifted along the sequence, this decreasing list of $s$-mers must eventually come to an end, hence there is a window guarantee.
This window guarantee is weak as the window can be as long as $\sigma^{k-1}$ (see Supplementary Material~\ref{sec:expon-wind-guar}).

If the first bit is not set, because of the left-most tie breaking rule, there is no window guarantee.
Hence, these methods have a short context and a weak or missing window guarantee.

% XXX: Mention low entropy sequences that are naturally skipped? This is
% interesting but we are running out of space and would take a fair bit of
% explanation and is only tangentially related.

\section{Cycle structure of the de~Bruijn graph}

% 1) Methods and their window guarantee: Rabin-Karp, Word order, syncmer, minimizers, fractional minimizers, minmers; DS remaining path length.
% 2) PCR sets and MDS, Mykkeltveit and Champarnaud
% 3) F-moves and I-moves, cycle signature preserved and modified
% 4) I-move conjecture
% 5) Finding I-moves, I-move signature conjecture, traversing graph
% 6) Remaining path length ranges, of Mykkeltveit and Champarnaud, polynomial conjecture
% 7) Discussions: range of remaining path length vs. alphabet size; diff I-move traversing algorithm and the conjecture; avoiding low-entropy sequences

There exists two methods to generate decycling sets of minimum size by Mykkeltveit~\cite{mykkeltveitgolomb} and Champarnaud~\cite{champarnaud}.
These algorithms are of great theoretical importance as they settled a conjecture of Golomb~\cite{golomb} on the size of an MDS.
They are also practical algorithms as membership in these MDSes is testable in time and memory polynomial in $k$ (i.e., the entire set does not need to be precomputed and stored).
But, as we shall see, the space of all MDSs is much larger than these two MDSs.

We provide a method that uses only two simple graph operations---called F-move and I-move---that transform an MDS into another MDS.
Furthermore, we conjecture that these two operations are sufficient to enumerate all MDSs.
In other words, given a graph where the nodes are all the MDSs and the edges represent these operations, Conjecture~\ref{conj:imove-conn} states that this graph is strongly connected.
We give theoretical and computational evidence to support this conjecture.

This section describes the structure of the cycles in the de~Bruijn and how through these two operation MDSs interact with the cycles.
Although these two operations are similar in nature and together they might enumerate all MDSs, we describe them separately as they have qualitatively distinct effects on the MDSs (see Proposition~\ref{thm:fmove-preserve-mds} and Conjecture~\ref{conj:imove-signature}).

A \emph{pure cycling register} (PCR), aka a \emph{conjugacy class}, is a cycle in the de~Bruijn graph made of the circular permutation of a \kmer.
For example, the PCR of the $4$-mer $1011$ over the binary alphabet is $1011 \rightarrow 0111 \rightarrow 1110 \rightarrow 1101 \rightarrow 1011$.
The PCRs form a partition of the \kmers and therefore any MDS must contain at least one \kmer from each PCR.
We call a \kmer set with exactly one \kmer in each PCR a \emph{PCR set}.
The theorems of Mykkeltveit~\cite{mykkeltveitgolomb} and Champarnaud~\cite{champarnaud} show that every MDS is a PCR set.
On the other hand, not every PCR set is an MDS.

\subsection{F-moves}

\begin{figure}
  \centering
  \includegraphics{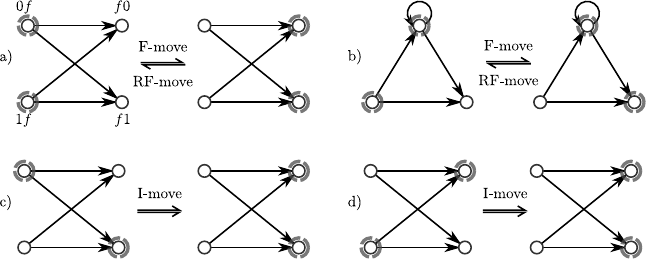}
  \caption{
    a)~For $f \in \Sigma^{k-1}$, the left-companions (\kmers $0f$ and $1f$ for the binary alphabet) and right-companions ($f0$ and $f1$) induce a directed complete bipartite $K_{\sigma,\sigma}$.
    When the left-companions are in the set (left subgraph, highlighted in gray), an F-move replaces these nodes with the right-companions (right subgraph).
    An RF-move is the reverse operation, replacing the right-companions with the left-companions.
    b)~When one \kmer is a homopolymer, the induced subgraph is slightly different, but the F-moves and RF-moves are defined analogously.
    c)~One of the possible I-move, $\im{f}{1}$, where a mixture of left- and right-companions are in the set.
    d)~The other possible I-move, $\im{f}{2}$.
    For any $f \in \Sigma^{k-1}$ there are $1$ F-move, $1$ RF-move and $2^{\sigma} - 2$ I-moves possible, unless $f$ is a homopolymer.
  }
  \label{fig:moves-companions}
\end{figure}

The \emph{left-companions} (resp\@. \emph{right-companions}) is the set of \kmers that have the same suffix (resp\@. prefix).
Given $f \in \Sigma^{k-1}$, then $\lc(f) \triangleq \{ af \mid a \in \Sigma \}$ are the left companions sharing the suffix $f$, and $\rc(f) \triangleq \{ fa \mid a \in \Sigma \}$ are the right companions.
See Figure~\ref{fig:moves-companions} for examples.
If $f=a^{k-1}$, then the \kmers $af$ and $fa$ are equal (homopolymer $a^{k}$), and this \kmer is both in the left- and right-companion sets for $f$.
The homopolymers are the only such \kmers.
Every other \kmer is a left-companion for exactly one suffix and a right-companion for a different prefix.

\begin{proposition}[Existence of F-moves]\label{thm:f-moves-exist}
  In any MDS $M$, there exists $f, f' \in \Sigma^{k-1}$ such that $M$ contains the left companions of $f$ and the right companions of $f'$.
\end{proposition}
\begin{proof}
  By contradiction, assume there is no such $f'$.
  Color all the nodes of the graph blue and do a random walk in the graph, starting from any node not in $M$, avoiding the nodes in $M$.
  Color in red the nodes traversed.
  Any \kmer $m$ is the left-companion of a suffix, say $f_{m}$, and every outgoing edge from $m$ is an incoming edge to a right-companion of $f_{m}$ (see Figure~\ref{fig:moves-companions}).
  Because no right-companion sets are in $M$, it is always possible to continue the walk avoiding $M$ from any $m$.
  Given that the graph is finite, the red nodes will eventually create a cycle, contradicting $M$ being a decycling set.
  The same reasoning applies for the existence of $f$ traversing edges in the reverse direction.
  \qed
\end{proof}

An \emph{F-move} (named after Fredricksen~\cite{fredricksen}) in $M$ for $f \in \Sigma^{k-1}$ is the operation of changing the set of left-companions of $f$ for the set of right-companions, as shown in Figure~\ref{fig:moves-companions}.
We use the functional notation $fM$ to designate the set obtained by the valid F-move $f$ from $M$: $fM \triangleq M \cup \rc(f) \setminus \lc(f)$.
This is a valid operation only when $M$ contains $\lc(f)$.
As a consequence of Proposition~\ref{thm:f-moves-exist} there always exists a valid F-move in an MDS.
The \emph{RF-move} (reverse F-move) is the inverse operation, valid when $M$ contains $\rc(f)$, $f^{r}M \triangleq M \cup \lc(f) \setminus \rc(f)$, satisfying $f^{r}fM = ff^{r}M = M$.

% TODO: CLK: wouldn't the f move operation be better written as $(M \setminus lc(f)) \cup \rc{f}$ ? I.e. the order of ops is: remove the lc, add the rc; only defined when the first succeeds. The way it's written now, there is some possibility that \cup rc allows the graph to include lc. Even though the text prohibits it (and this proposed reordering is equivalent to the text), I feel like putting the "escape hatch" at the same place as the order of operations is clearer. It also avoids the non-universal? assumption that \ doesn't have higher precedence than \cup [M - rc*lc], ditto for f^r.

\begin{proposition}[F-moves preserve decycling sets]\label{thm:fmove-preserve-mds}
  Let $M$ be an MDS such that $\lc(f) \subset M$, then $fM$ is also an MDS.
\end{proposition}
\begin{proof}
  If there is a cycle that avoids $fM$, then it must use one of the nodes in $\lc(f)$, otherwise it was already a cycle avoiding $M$.
  Any cycle using a node in $\lc(f)$ then must use a node in $\rc(f) \subset fM$.
  \qed
\end{proof}

An analogous statement holds for RF-moves.
F-moves give a procedure to enumerate some MDSs, starting for example from either the Mykkeltveit or Champarnaud set and repeatedly applying a (guaranteed-to-exist by Prop.~\ref{thm:f-moves-exist}) F-move.
Unfortunately, not all MDSs are reachable using only F-moves.
The \emph{MDS graph} $\gmds(\sigma, k)$ has all the MDSs as nodes and edges that represent F-moves operations between MDSs.
$\gmds$ is not connected, as seen in Figure~\ref{fig:mds-graph}, but its components have a well characterized structure (proof in Supplementary Material~\ref{sec:mds-structure}).

\begin{restatable}[$\gmds$ component structure]{proposition}{mdsgraphcomponents}\label{thm:comp-structure}
  For any $\sigma$ and $k$, the components of $\gmds(\sigma, k)$ satisfy:
  \begin{enumerate}
    \item \label{item:strongly-connected} every component is strongly connected
    \item \label{item:cycle-length} every cycle is of length $\alpha \sigma^{k-1}, \alpha \in \mathbb{N}$
    \item \label{item:fmoves-cycles} in a cycle of length $\alpha \sigma^{k-1}$, every possible F-move $f \in \Sigma^{k-1}$ occurs exactly $\alpha$ times
    \item \label{item:nodes-cycles} every node is in a cycle of length $\sigma^{k-1}$ (hence the girth is $\sigma^{k-1}$)
    \item \label{item:thickle} each component is a $\sigma^{k-1}$-partite directed graph
   \end{enumerate}
\end{restatable}

\begin{figure}
  \centering
  \includegraphics{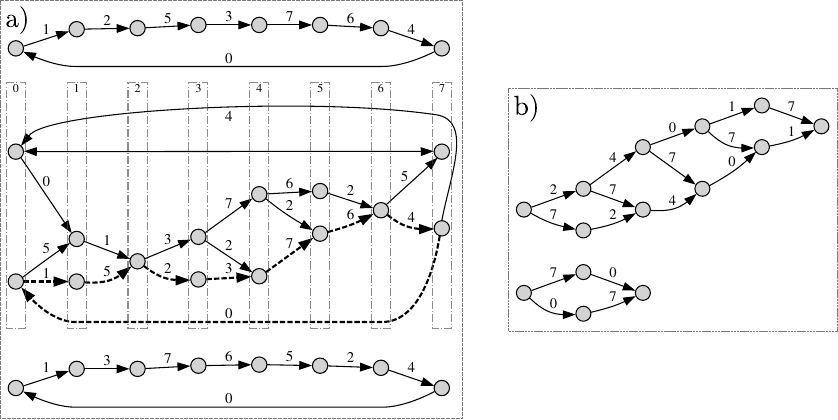}
  \caption{a)~MDS graph $\gmds(2, 4)$ with edge labels as numbers in $[0, \sigma^{k-1}]$ representing the F-moves.
    There are 3 components.
    Each component is strongly connected and can be partitioned into $\sigma^{k-1}=8$ layers with edges only from one layer to the next.
    The gray vertical boxes in the middle component highlight the layers, numbered from $0$ to $\sigma^{k-1}$.
    Each layer in the middle component has size 1 or 2.
    An example of a cycle of length $8$ with every F-move done exactly once is highlighted with dashed edges.
    b)~Example of $2$ components of non-decycling PCR sets.
    The components are DAGs with a longest path less than $8$ edges.
  }
  \label{fig:mds-graph}
\end{figure}

\subsection{I-moves}

An I-move, as in an ``incomplete F-move'', is valid when $M$ contains a mixture of left- and right-companions: for some $f \in \Sigma^{k-1}$ and $\forall a \in \Sigma$, either $af$ or $fa$ is in $M$.
See Figure~\ref{fig:moves-companions} for an example.
For a given $f \in \Sigma^{k-1}$, there are $2^{\sigma}-2$ distinct I-moves: one for each possible choice of left-companions nodes in $M$, excluding the F-move (all of $\lc(f)$) and the RF-move (none of $\lc(f)$).
There is one exception: when $f = a^{k-1}$ is a homopolymer, $af = fa$ is both in $\lc(f)$ and $\rc(f)$ and the number of possible I-moves for $f$ is $2^{\sigma-1}-2$.

An I-move is denoted by $\im{f}{m}$ where $m \in [1, 2^{\sigma}-2]$ is interpreted as a bit-mask giving the nodes from $\lc(f)$ (i.e., the $a$th bit $m_{a} = 1$ iff $af \in M$ and $m_{b} = 0$ iff $fb \in M$).
With this notation, the F-move $f$ is equivalent to $\im{f}{0}$ while the RF-move is $\im{f}{2^{\sigma}-1}$.
An identical argument as for Proposition~\ref{thm:fmove-preserve-mds} shows that applying a valid I-move to an MDS also gives an MDS

Although F-moves and I-moves seem like similar operations and both preserve MDSs, they have distinct effect on MDSs.
First, empirically we observe that I-moves, unlike F-moves, are not always possible.
MDSs always have a valid F-move (Proposition~\ref{thm:f-moves-exist}), while an MDS may not have any valid I-move.
All of the $\sigma^{k-1}$ F-moves are an edge in every component of the MDS graph, while not all of the $\sigma^{k-1} \cdot (2^{\sigma}-2)$ I-moves are valid in at least one MDS of the entire MDS graph.
In particular, no MDS for $\sigma=2$ and $k=5$ have any valid I-move.

% TODO: CLK: "while not all of the $..$ I-moves are valid in at least one MDS" => "while some of the I-moves are invalid in at least one MDS of the entire MDS graph".

% TODO: CLK: "entire MDS graph" has this been defined? is it the same as G_{MDS}?

Second, F-moves not only preserve the decycling property of MDSs, but they also preserve the ``coverage'' of every cycle by an MDS.
To make this notion precise, define the \emph{hitting number} of a cycle $C$ of $D_{k}$ by the MDS $M$ as the size of their intersection: $\hit_M(C) = |M \cap C|$.
Because $M$ is a decycling set, necessarily $\hit_{M}(C) \ge 1$.
PCRs for example have a hitting number of $1$ while any Hamiltonian cycle has a hitting number equal to $|M|$.

Furthermore, the \emph{cycle signature} of MDS $M$ is the vector of all hitting numbers for all possible cycles: $\sig(M) = \bigl\langle\hit_{M}(C)\bigr\rangle_{C \textrm{ cycle of } D_{k}}$.
Per the following proposition, F-moves preserve hitting numbers and signatures, while I-moves do not.

\begin{restatable}{proposition}{signatures}
  \begin{enumerate}
    \item\label{item:hit-fmove} Let $M$ be an MDS and $f$ a valid F-move in $M$, then for any cycle $C$, $\hit_{M}(C) = \hit_{fM}(C)$
    \item\label{item:hit-imove} For every valid I-move $\im{f}{m}$ in MDS $M$, there exists a cycle $C$ of $D_{k}$ such that $\hit_{M}(C) \ne \hit_{\im{f}{m}M}(C)$
    \item\label{item:sig-same} For any MDSes $M_{1}, M_{2}$ from the same component of $\gmds$, $\sig(M_{1}) = \sig(M_{2})$
    \item\label{item:sig-diff} For any MDSes $M_{1}, M_{2}$ from different components of \gmds, $\sig(M_{1}) \ne \sig(M_{2})$
  \end{enumerate}
\end{restatable}
\begin{proof}
  Let $f$ be a valid F-move in MDS $M$, and $C$ be a cycle of $D_{k}$.
  Because every outgoing edge of a node in $\lc(f)$ is an incoming edge to a node in $\rc(f)$, $C$ must contain as many nodes from $\lc(f)$ as from $\rc(f)$ (which can be $0$).
  Before the F-move, all the nodes from $\lc(f)$ and none from $\rc(f)$ are in $M$, while the opposite is true for $fM$.
  Hence the hitting number is unaffected by the F-move, proving~\ref{item:hit-fmove}.

  Let $\im{f}{m}$ be a valid I-move in $M$ such that $af \in \lc(f)$ and $fb \in \rc(f)$, $a, b \in \Sigma$.
  Because $D_{k}$ is $(\sigma-1)$-vertex-connected~\cite{connectivity_dbg}, there exists a path $P$ from $fb$ to $af$ that avoids $cf, c \in \Sigma \setminus \{a\}$.
  Path $P$ followed by edge $af \rightarrow fb$ form a cycle $C$ such that $\hit_{M}(C) = \hit_{fM}(C) + 1$ ($af$ is in $M$ but not in $fM$).
  By the same construction, there exists a ``complementary'' cycle $C'$ using $bf$ and $fa$ such that $\hit_{M}(C') = \hit_{fM}(C') - 1$. % G: Extra information. Keep? Remove?
  This proves~\ref{item:hit-imove}.

  As a component of \gmds is strongly connected by F-moves, statement~\ref{item:sig-same} is a direct consequence of~\ref{item:hit-fmove}.
  A proof for~\ref{item:sig-diff} is given in Supplementary Material~\ref{sec:cycle-sign-uniq}.
  \qed
\end{proof}

As a consequence of this proposition, the hitting number and signature are constant over a component of the MDS graph, and the hitting number $\hit_{\chi}(C)$ and the signature $\sig(\chi)$ are well defined for a component $\chi$.
Because an I-move changes the signature, every I-move links MDSs from different components.
Consider now the \emph{component graph} $\gcomp(\sigma, k)$ with one node for each component of \gmds and a directed edge from component $\chi_{1} \rightarrow \chi_{2}$ if there is an I-move from an MDS $M_{1} \in \chi_{1}$ to $M_{2} \in \chi_{2}$.
In fact, as stated in the following Proposition, \gcomp is an undirected graph (proof in Supplementary Material~\ref{sec:gcomp-undirected}).

\begin{restatable}[$\gcomp$ is undirected]{proposition}{complementimove}
  Let $\im{f}{m}$ be a valid I-move from MDS $M_{1}$ in component $\chi_{1}$ to $M_{2}$ in $\chi_{2}$.
  Then there exists $M'_{2}, M'_{1}$ in $\chi_{2}, \chi_{1}$, respectively, such that $\im{f}{\overline{m}}$ (where $\overline{m}$ is the bit-complement of $m$) is a valid I-move from $M'_{2}$ to $M'_{1}$.
\end{restatable}

\subsection{Enumerating all MDSs}

We make the following two conjectures regarding the use of I-moves to enumerate all MDSs.

\begin{conjecture}[Connectivity by I-moves]\label{conj:imove-conn}
  The \gcomp graph is connected.
  Equivalently, every MDSs is reachable from the Mykkeltveit MDS using a sequence of F-moves and I-moves.
\end{conjecture}

This conjecture is supported by the previous theoretical results, in particular that all the components have a different signature and that the I-move always change the signatures.
For reasonable values of $k$ ($\sigma=2$, $k \le 7$), it is computationally feasible to enumerate all PCR sets and check which of them are also decycling sets.
Using this brute force method we can confirm that $\gcomp(2, k)$ is connected up to $k=7$.

The following conjecture is also verified up to $k=7$ and exposes another fundamental difference between F-moves and I-moves.
Every F-move is always valid in every component, while the valid I-moves identify a component (similarly to the cycle signature).
For a component $\chi$, let the list of I-moves be $\sigim(\chi) = \{ \im{f}{m} \mid \exists M \in \chi \text{ where } \im{f}{m} \textrm{ is a valid I-move in } M \}$.

\begin{conjecture}[I-move signature]\label{conj:imove-signature}
  Every component in $\gmds$ has a distinct list of valid I-moves.
\end{conjecture}

The validity of this second conjecture is likely related to the previous one.
To prove Conjecture~\ref{conj:imove-conn}, one needs to show that for any two components $\chi_{1}, \chi_{{2}}$ there is a path of I-moves to go from $\chi_{1}$ to $\chi_{2}$.
Conjecture~\ref{conj:imove-signature} can be used as a guide to find that path: because $\sigim(\chi_{1}) \ne \sigim(\chi_{2})$, then there exists a valid I-move in either $\sigim(\chi_{1}) \setminus \sigim(\chi_{2})$ or $\sigim(\chi_{2}) \setminus \sigim(\chi_{1})$.
(Note that it is possible to have, for example, $\sigim(\chi_{1}) \subset \sigim(\chi_2)$.)
Do that I-move and repeat with the new components.
Although in our testing Conjecture~\ref{conj:imove-signature} is useful to find a path from $\chi_{1}$ to $\chi_{2}$, it is not sufficient as it does not guarantee that the size of the difference between the I-move lists is decreasing.

To create Table~\ref{tab:number-comps} we use both conjectures: one to traverse the graph and the other to avoid enumerating a component more than once.

\begin{table}
  \caption{\gcomp and \gmds properties for $\sigma=2$.
    % The number of MDSs and components seems lower when $k$ is prime than when $k$ is composite.
    % There are no formula to compute these numbers.
    ``Layer range'' gives, when possible, the range of the number of MDSs in each layer of \gmds.
    The numbers for $k \le 7$ are exact, computed from the exhaustive list of MDSs.
    For columns $k \in [8, 10]$, the number of components is correct provided the conjectures are correct, otherwise the numbers provided are under-estimations.
    For $k = 8$, the layer size and number of MDSs are estimated by sampling $100$ random components.
    For $k = 9$, the numbers are likely severe under-estimations.
    For $k = 10$, computation is too expansive.
  }
  \label{tab:number-comps}
  \centering
  \ifllncs
    % using llncs
  \else
    \smaller
  \fi
  \setlength{\tabcolsep}{6pt}
  \begin{tabular}{lrrrrrrrrr}
    \toprule
    Method        & \multicolumn{6}{c}{Exhaustive} & \multicolumn{3}{c}{I-moves}                                                                                                                                                      \\
    \cmidrule(l){2-7} \cmidrule(l){8-10}
    $k$           & 2                              & 3                & 4                & 5                & 6                 & 7                   & 8                                        & 9                & 10              \\
    \midrule
    \# components & 1                              & 1                & 3                & 1                & 273               & 4                   & \num{194133}                             & \num{4318173}    & \num{195740496} \\
    \# MDSs       & 2                              & 4                & 30               & 28               & \num{68288}       & \num{18432}         & $\approx \num{3.1e11}$                   & $> \num{1.3e17}$ &      ---           \\
    Layer range   & \num{1}--\num{1}               & \num{1}--\num{1} & \num{1}--\num{2} & \num{1}--\num{2} & \num{1}--\num{48} & \num{28}--\num{153} & $\approx \num{2.5e3}\text{--}\num{29e3}$ & $> \num{1.2e8}$  &     ---            \\
    \bottomrule
  \end{tabular}
\end{table}

\subsection{Non-decycling PCR sets.}

Non-decycling PCR sets may also have valid F-moves and I-moves, but there are significant differences with MDSs.
Unlike MDSs (see Proposition~\ref{thm:f-moves-exist}), a non-decycling set it is not guaranteed to contain sets of left- and right-companions.
Even more, the analog graph to \gmds with non-decycling PCR sets as nodes and F-moves for edges is a non-connected graph where each component is a DAG (see Figure~\ref{fig:mds-graph} and Supplemental Material~\ref{sec:non-decycling-pcr}).
There cannot be any F-moves between an MDS and a non-decycling set.
On the other hand, there can be an I-move from a non-decycling set to an MDS (but not the other way around).

\section{Remaining path length and window guarantee}

By traversing the component graphs and the MDS graph, one can search for MDSs with desirable properties.
Unfortunately, as seen in Table~\ref{tab:number-comps}, every aspect of these graphs (i.e., number of MDS, number of components, layer size, etc.) seem to have super-exponential growth.
Enumerating all MDSs for $k \ge 9$ with the binary alphabet is likely not reasonable, and for the DNA alphabet it is even more difficult.
In this section, we provide some methods to explore the space of MDSs more efficiently and study the window guarantee of MDSs.

\subsection{Efficiently traversing the component graph.}\label{sec:effic-trav-comp}

As is seen in Table~\ref{tab:number-comps}, the number of MDSs and components is increasing quickly with $k$, although an actual estimate of the growth as a function of $k$ is not known.
The memory used to traverse a component can be reduced by noticing that each component is partitioned into $\sigma^{k-1}$ layers with edges only from one layer to the next (see Figure~\ref{fig:mds-graph}).
Therefore, it is only necessary to keep in memory the MDSs of the current and next layer to exhaustively enumerate every MDSs in the component.

As each component contains at least one cycle of length $\sigma^{k-1}$, the number of MDSs grows by at least a factor of $\sigma^{k-1}$ faster than that of components.
In fact, it grows much faster as each of the $\sigma^{k-1}$ layers has a size that grows fast with $k$ as well (see Table~\ref{tab:number-comps}).
While the number of MDSs and the size of the layers varies significantly between components, in general it is not efficient to traverse an entire component to find all the valid I-moves.
Using the following proposition, it is possible to find all the valid I-moves in a component by considering only one MDS.

Given an MDS $M$, any cycle $C$ satisfies $H_{M}(C) \ge 1$.
The cycles with a hitting number of exactly $1$, called \emph{constrained cycles}, play an important role in the existence or not of a valid I-move: an I-move is only valid if there is no constrained cycle using edges of the I-move.

\begin{restatable}{proposition}{imoveconstrained}
  Let $f \in \Sigma^{k-1}, m \in [1, 2^{\sigma}-2]$, and let $\chi$ be a component of $\gmds$.
  Then $\im{f}{m}$ is not a valid I-move in any MDS of $\chi$ if and only if $\exists a, b$ such that $m_{a} = 1, m_{b} = 0$ and there exist a constrained cycle using the edge $af \rightarrow fb$.
\end{restatable}

This proposition, proved in Supplementary Material~\ref{sec:i-move-constrained}, shows that to find the list of valid I-moves in the entire component it is sufficient to find the edges not covered by a constrained cycle in just one of the MDS of the component.
This holds, as by Proposition~\ref{item:hit-fmove}, the list of constrained cycles is constant across the MDSs of a component.
Moreover, tagging the edges covered by constrained cycle can be done with one depth-first search for each \kmer in the MDS.
The main advantage of this method is its run time is independent of the number of MDSs in the component.
% G: this is not as bad as it may seem as the length of DFS is bounded by the remaining path length. Which is likely polynomial and much less that \sigma^k.

\subsection{Remaining path length}\label{sec:rema-path-length}

The \emph{remaining path length} of an MDS $M$ is the length of the longest path in the DAG obtained by removing the \kmers of $M$ from $D_{k}$.
Given a selection scheme that selects in a sequence the \kmers from $M$, the remaining path length is precisely the window guarantee of the scheme.
The following proposition gives bounds on the effect of an F-move or I-move on the remaining path length (see Figure~\ref{fig:path-change}).

\begin{proposition}\label{thm:change-path-len}
  An F-move or RF-move can increase or decrease the remaining path length by at most $1$.
  An I-move can increase the remaining path length by at most $1$ or decrease it by half.
\end{proposition}
\begin{proof}
  First, notice that the longest path in $D_{k} \setminus M$ must start at a valid F-move and end at a valid RF-move.
  Let $P = (m_{1}, \ldots, m_{n})$ be a longest path.
  The \kmer $m_{1}$ is the right-companion of some suffix $f$.
  Suppose there exists $a \in \Sigma$ such that $af \notin M$, then the path $P' = (af, m_{1}, \ldots, m_{n})$ avoids $M$ and is longer than $P$, contradicting its maximality.
  Therefore $\lc(f) \subset M$ and $f$ is a valid F-move in $M$.
  The proof is symmetrical for $m_{n}$ as the left-companion of some prefix $f'$ with $\rc(f') \subset M$.

  Because $m_{1} \in fM$, the path $P$ is shortened by $1$ by the F-move $f$, which may shorten the longest path if there was no other paths of that length.
  Also, $\rc(f) \subset fM$ (i.e., $f$ is a valid RF-move in $fM$ but it was not in $M$), hence there might be maximal path $P'$ ending at a left-companion of $f$ with $|P'|>n$.
  Because the F-move only moved nodes forward by one edge, $|P'| \le n+1$ and the longest path may have increased by $1$.
  The same argument applies to an RF-move.

  For a valid I-move $\im{f\dprime}{m}$ in $M$, the same reasoning applies for increasing by $1$.
  On the other hand, a longest path may have used an edge $af\dprime \rightarrow f\dprime{}b$ where $m_{a} = 0, m_{b} = 1$.
  That is $P = (m_{1}, \ldots, m_{i} = af\dprime, m_{i+1} = f\dprime{}b, \ldots, m_{n})$.
  After the I-move, $fb \in \im{f\dprime}{m}M$ and the path is now broken in up to two parts: $(m_{1}, \ldots, m_{i})$ and $(m_{i+2}, \ldots, m_{n})$.
  Therefore the remaining path length could be halved if $i = n/2$.
  \qed
\end{proof}

Based on this, we implemented a simulated annealing algorithm to find the smallest and largest remaining path lengths among MDSs.
The longest path for the MDS $M$ is computed using a modified topological sort of the DAG $D_{k} \setminus M$.
Supposed we are computing the smallest remaining path length.
Starting from a component of the MDS graph, the program performs a fixed number of random F-moves ($2k$ by default) and computes the remaining path length for each MDS and keeps the minimum.
Then, it finds all the valid I-moves in the current component as explained in Section~\ref{sec:effic-trav-comp}, and it picks one at random.

After performing the I-move, in the new component, the remaining path length is computed for $2k$ MDSs reachable by F-moves and a new minimum is computed.
If this new minimum is lower than the previous minimum, then the new component becomes the current component.
Otherwise, it becomes the current component only with some small probability.
Then the process is repeated from the current component for a fixed number of iterations.
As is traditional with simulated annealing, the probability to jump to ``worse'' components decreases over time.

\begin{table}
  \setlength{\tabcolsep}{3pt}
  \renewcommand{\u}{\underline}
  \caption{The remaining path length for the Mykkeltveit and Champarnaud sets compared to the range of remaining path length.
    For $\sigma=2$ and $k \le 7$ (underscored), the range of remaining path length is computed exactly from the exhaustive list of MDSs.
    All other values are estimated using a simulated annealing (SA) algorithm.
  }
  \label{tab:remaining-path-len}
  \centering
  \ifllncs
    % llncs
  \else
    \smaller
  \fi
  \begin{tabular}{clrrrrrrrrrrrrrrrrr}
    \toprule
    $\sigma$           & Algorithm   & \multicolumn{17}{c}{$k$}                                                                                           \\
    \cmidrule(l){3-19}
                       &             & 4     & 5      & 6      & 7      & 8   & 9   & 10   & 11  & 12  & 13  & 14  & 15  & 16  & 17  & 18   & 19   & 20   \\
    \midrule
    \multirow{4}{*}{2} & Mykkeltveit & 5     & 11     & 21     & 27     & 39  & 55  & 74   & 89  & 119 & 143 & 194 & 219 & 253 & 299 & 408  & 437  & 539  \\
                       & Champarnaud & 7     & 11     & 21     & 27     & 47  & 57  & 94   & 112 & 190 & 209 & 367 & 415 & 683 & 756 & 1343 & 1393 & 2560 \\
                       & SA Min      & \u{5} & \u{11} & \u{13} & \u{25} & 32  & 48  & 70   & 89  & 119 & 143 & 194                                        \\
                       & SA Max      & \u{7} & \u{12} & \u{26} & \u{32} & 55  & 80  & 116  & 158 & 257 & 288 & 387                                        \\
    \midrule
    \multirow{4}{*}{4} & Mykkeltveit & 21    & 41     & 77     & 111    & 145 & 231 & 330  & 403  & 616                                                   \\
                       & Champarnaud & 27    & 39     & 119    & 141    & 429 & 520 & 1601 & 1765 & 6180                                                    \\
                       & SA Min      & 20    & 41     & 77     & 111    & 145                                                                             \\
                       & SA Max      & 34    & 66     & 149    & 270    & 530                                                                             \\
    \bottomrule
  \end{tabular}
\end{table}

Table~\ref{tab:remaining-path-len} shows the remaining path length for the two previously known algorithms to generate MDSs and the range of remaining path length.
These ranges are either exact when an exhaustive list of MDSs is computable, and approximated using simulated annealing otherwise.
Based on the pattern that the Mykkeltveit set is always at or close to the minimum remaining path length, we conjecture that it holds for all parameter $k$ and $\sigma$.

\begin{conjecture}\label{conj:mykk-rema-path-len}
  For a given $\sigma$, let $\ell_{\textrm{min}}(k), \ell_{\textrm{max}}(k), \ell_{\textrm{Mykk}}(k)$ respectively be the smallest, largest and Mykkeltveit set remaining path lengths.
  Then $\ell_{\textrm{Mykk}}(k) - \ell_{\textrm{min}}(k) = o(\ell_{\textrm{max}}(k) - \ell_{\textrm{min}}(k))$ asymptotically in $k$.
\end{conjecture}

\subsection{Per-component remaining path length}

Proposition~\ref{thm:change-path-len} gives a bound to the change in the remaining path length as the MDS graph is traversed using F-moves and I-moves.
Within one component, given that every MDS is in a cycle of length $\sigma^{k}$, the remaining path length along this cycle could change by up to $\sigma^{k}/2$.
In other words, this proposition only gives an exponential bound on the range of remaining path length within a component.

The graph in Figure~\ref{fig:path-change} has a point for each component at the coordinate $(m_{P}(\chi), M_{P}(\chi))$ where $m_{P}(\chi)$ is the minimum of the remaining path length over all the MDSs of the component $\chi$, and $M_{P}(\chi)$ is the maximum.
The vertical distance from the diagonal $y = x$ represents the range of remaining path lengths within a component.
We observe for $k \le 8$ on the binary alphabet that the range is bounded by $O(k)$.

\begin{conjecture}\label{conj:lin-comp-len}
  Within a component of \gmds, the range of remaining path length is $O(k)$.
\end{conjecture}

There are plausible reasons for having such a small range.
Consider two extremes: (1)~there are many F-moves and RF-moves valid at the same time in an MDS $M$, (2)~there is only 1 F-move and 1 RF-move valid in $M$.
In the first case, doing one of these F-moves or RF-moves affects the maximal paths that start or end at these moves.
Consequently, many of these moves change the length of paths that are not the longest.
In other words, these moves have no effect on the remaining path length.
In the second case, it is possible to show that doing the 1 valid F-move does not change the remaining path length (the longest path is truncated by its first node and extending by one node, hence not changing in length).
This type of situation is likely to happen when there are few F-moves and RF-moves possible.
In both cases, most F-moves do not affect the remaining path length.

This conjecture partially justifies only exploring $O(k)$ MDSs within one component in the simulated annealing algorithm in Section~\ref{sec:rema-path-length}.

\begin{figure}
\begin{minipage}{0.49\textwidth}
  \centering
  \includegraphics[scale=1.2]{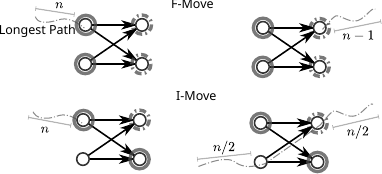}
\end{minipage} \hfill
\begin{minipage}{0.49\textwidth}
  \centering
  \includegraphics{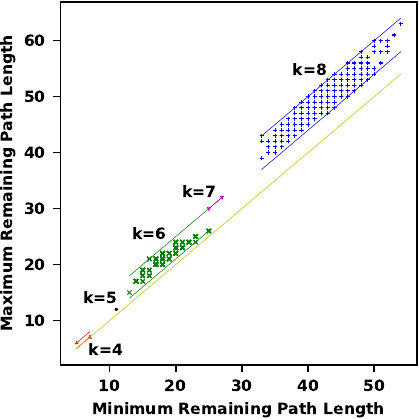}
\end{minipage}
\caption{Left: If a longest path does not start at a valid F-move $f$, i.e., one of the left-companion of $f$ in solid gray is missing, then it could be extended to the left, contradicting maximality.
  Doing F-move $f$ (changing solid gray for dashed nodes) can shorten the longest path by $1$ node.
  Also, after doing F-move $f$, a path now ending in one of the solid gray node could be the longest and was extended by $1$ node.
  If the path goes through an I-move $\im{f\dprime}{m}$, then doing the I-move cuts the path in two possibly equal parts.
  Right: Comparison of the minimum and maximum remaining longest path for components of $\gmds(2,k)$ for $4 \le k \le 8$.
  Each point represents one connected component of the graph.
  The minimum and maximum remaining path lengths are computed over all the MDSs of a component.
  Therefore, the vertical distance of a point from the diagonal $y=x$ (in yellow) shows the variation of remaining path length within a component.
  For $k=8$, a subsample of 500 components were examined, as the total number of components is exceedingly large.
  The lines are drawn to depict the bounds of the increase between components.
  In all cases seen, the difference between the minimum and maximum remaining length within a component is in some range $[\alpha,\alpha+k]$ for an alpha that is less than $k$.
  }
  \label{fig:path-change}
\end{figure}

\section{Discussion}

\paragraph{Proportion of MDSs.}
A simple algorithm to generate a random MDS, sampling the space of MDSs uniformly, is to select at random \kmer from each PCR and check whether it is decycling, and to resample if not.
Even though the space of MDSs is (maybe surprisingly) large, it is nonetheless only a tiny fraction of the PCR sets.
The number of PCR sets is easily computable~\cite{necklaces} and asymptotically there are $\Omega(k^{\sigma^{k}/k})$ PCR sets.
There is no formula for the number of MDSs, but based on the numbers from Table~\ref{tab:number-comps}, for $k=8$ %I rearranged this but still not sure I like it
of the $\num{2e29}$ PCR sets
the proportion that are MDSs is only $\num{2e-18}$. %I know proportion is technically correct but I think fraction makes this more obvious this isn't a typo
For $k=9$ %the number of PCR sets is $\num{1.7e54}$ and the 
that proportion 
%of MDSs 
is essentially~$0$.
Thus, the random
%This 
sampling method is not of any practical use.

In that sense Conjecture~\ref{conj:imove-conn}, provided it is true, is an efficient method to enumerate all MDSs as only MDSs are ever considered without the need to filter out an overwhelming number non-decycling sets.
Even if this conjecture is eventually proven wrong, the F-moves and I-moves allow us to explore a large subspace of MDSs, and, using simulated annealing or more advanced machine learning methods, to find MDSs with desirable properties.

Moreover, on the theoretical side, providing evidence for this conjecture lead us to a deeper understanding of the space of MDSs and to formulate other useful conjectures.

% The simulated annealing method of Section~\ref{sec:rema-path-length} is relatively simplistic.
% Need to say something about it.

\paragraph{Mykkeltveit set and short windows.}
It is surprising (or lucky) that the first algorithm for constructing MDSs by Mykkeltveit~\cite{mykkeltveitgolomb} gives a set with close to the shortest remaining path length.
This fact may explain retrospectively the success of previous methods using this set as the starting point to design minimizers schemes~\cite{DOCKS,DOCKS2,Pasha,EfficientMinimizerOrders}.
The growth of the remaining path length for the Mykkeltveit set is well characterized~\cite{MykkBounds}: it is $\Omega(k^{2})$ and $O(k^{3})$.
Fitting the data from Table~\ref{tab:remaining-path-len} we obtain an exponent of $3.12 \pm 0.14$, suggesting an actual growth of $O(k^{3})$.
Provided that Conjecture~\ref{conj:mykk-rema-path-len} holds, this would answer the question of the shortest window guarantee that is possible using an MDS.
For comparison, fitting the Champarnaud data gives an exponent of $6.1 \pm 0.59$.

\paragraph{Longest remaining path length.}

Conjecture~\ref{conj:lin-comp-len} only suggests a bound on the range of remaining path length within a component of \gmds.
A legitimate question is what is the bound of the range in \gmds as a whole.
Figure~\ref{fig:path-change} could suggests that this range is polynomial in $k$, although the trend in this figure is much too short to elevate this statement to a conjecture.
Given the known results bounding the longest remaining path of the Mykkeltveit set by $O(k^{3})$, this would mean a polynomial bound on the remaining path length of MDSs.

This statement seems counterintuitive at first (and is, of course, not proven).
We saw in Section~\ref{sec:wind-guar-exist} that syncmers have a window guarantee of $\sigma^{k-1}$, hence there exists DSs that are not of minimum size that have exponentially long remaining paths.
How then can sets with fewer \kmers (MDSs) have a shorter remaining path length?
The intuition is as follows.
In the syncmers construction, we chose one exponentially long path (length $\sigma^{k-1}-1$) through the graph while every node not on this path is added to the DS $M$.
The size of the DS $|M| = \sigma^{k}(1-1/\sigma)$ is exponential as well: it takes many nodes, guiding that long path, to prevent cycles.
On the other hand, the size of an MDS is $\thicksim \sigma^{k}/k$, which is $o(\sigma^{k})$.
The average remaining path length is $k$ and there are too few \kmers in an MDS to guide an exponentially long path to prevent it from creating cycles (i.e., to have back edges).

% \paragraph{Window length dependence on $\sigma$}
% Other subjects of discussion?

\section{Conclusion}

The window guarantee is an important requirement, theoretically and practically, to define and optimize sketching methods.
As discussed, an underlying concept that can be extracted from the definition of this guarantee in any local sketching method is a set of nodes in the de~Brujin graph which are unavoidable (i.e., decycling).
While many such sets exist, the minimum-sized sets have important properties that can be exploited and examined. 
In this work, we described some of the first theoretical findings on properties of these sets, as well as a method to traverse many (if not all) MDSs for a given \kmer length.
We also showed that the choice of MDS, whether direct or as an implication of the design of the sketching method, does have an impact on the strength of the window guarantee.
Although we provide our major results as conjecture, we present significant evidence to support these claims.

\noindent\textbf{Acknowledgement:}
This work was supported in part by the US National Science Foundation [DBI-1937540, III-2232121], the US National Institutes of Health [R01HG012470] and by the generosity of Eric and Wendy Schmidt by recommendation of the Schmidt Futures program.

\noindent\textbf{Conflict of Interest:}
C.K.\ is a co-founder of Ocean Genomics, Inc; G.M.\ is VP of software engineering at Ocean Genomics, Inc.

%%%%%%%%%%%%%%%%%%%%%%%%%%%%%%%%%%%%%%%%%%%%%%%%%%%%%%%%%%%%%%%%%%%%%%%%%%%%%%%%
% references
%%%%%%%%%%%%%%%%%%%%%%%%%%%%%%%%%%%%%%%%%%%%%%%%%%%%%%%%%%%%%%%%%%%%%%%%%%%%%%%%

\bibliographystyle{splncs04}
\bibliography{main}

%%%%%%%%%%%%%%%%%%%%%%%%%%%%%%%%%%%%%%%%%%%%%%%%%%%%%%%%%%%%%%%%%%%%%%%%%%%%%%%%
%% Supplementary
%%%%%%%%%%%%%%%%%%%%%%%%%%%%%%%%%%%%%%%%%%%%%%%%%%%%%%%%%%%%%%%%%%%%%%%%%%%%%%%%
\clearpage\null\thispagestyle{empty}
%\newpage
\setcounter{page}{1}
\setcounter{section}{0}

\centerline{\huge \textbf{Supplementary Material}}
\vskip1in

\section{Exponential window guarantee for Decycling Sets (DS)}\label{sec:expon-wind-guar}

Consider a syncmer sketching method selecting a \kmer if the smallest $s$-mer is at position $1$ (first position).
Assume $s \le k-1$.
The order on the $s$-mer is as follows: create a de~Bruijn sequence $D$ of order $s$ (it contains all the $s$-mers once and only once) and $s_{1} < s_{2}$ iff the $s$-mer $s_{1}$ appears after $s_{2}$ in $D$.
The sequence $D$ is a decreasing sequence of $s$-mers of length $\sigma^{s}+s-1$.
With $s = k - 1$, we created a sequence of length $\Omega(\sigma^{k-1})$ without a selected \kmer.

\section{MDS graph structure}\label{sec:mds-structure}

\begin{figure}
  \centering
  \includegraphics[scale=0.8]{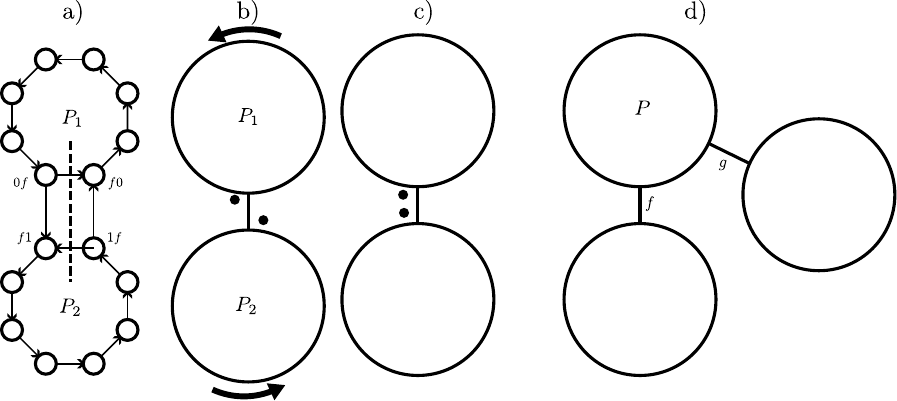}

  \caption{Simplified representation of PCRs, F-moves and I-moves when $\sigma = 2$.
    a)~shows two PCRs from the de~Bruijn graph $D_{k}$.
    Every \kmer is a circle, and they are all oriented counter-clock-wise (see PCR $P_{1}$ and $P_{2}$ here).
    Let $f$ be an F-move that involves $P_{1}, P_{2}$.
    Here $P_{1}$ has the edge $0f \rightarrow f0$, and $P_{2}$ has $1f \rightarrow f1$: these are the PCR edges.
    The cross-PCR edges $0f \rightarrow f1$ and $1f \rightarrow f0$ form anti-parallel edges between $P_{1}$ and $P_{2}$.
    b)~The \emph{simplified PCR/pebbles} representation shows PCRs as large cycles without representing individual \kmers and only representing the F-move edges of interest.
    The elements from the MDS in each PCR (the pebbles) are small black circles that can travel only counter-clock-wise around the PCR.
    An F-move is an edge between $P_{1}$ and $P_{2}$ and act as a semaphore: a pebble can move one step around the PCR and across the edge of $f$ only when the other pebbles are present next to the edge in the other PCR (i.e., $\lc(f)$ is in the MDS), as shown in b), and all pebbles move across the edge at the same time.
    c)~The position of the pebbles for the I-move $\im{f}{1}$: bit $0$ is set but not bit $1$, so the pebbles are on $0f$ and $f1$ (left side of the edge of $f$).
    The top pebble can move across the edge, counter-clock-wise, while the lower one stays still.
    For I-move $\im{f}{2}$ with bit $0$ unset and bit $1$ set, the pebbles would be on $1f$ and $f0$, on the right side of the edge of $f$.
    d)~If F-moves $f$ and $g$ have a PCR $P$ in common, then, because F-moves act like semaphores, it is not possible to do the F-move $f$ twice before $g$ is done once.
    For the pebble to go around $P$ to do $f$ a second time, necessarily the F-move $g$ was done as well.
  }
  \label{fig:representation}
\end{figure}

\begin{lemma}[Commutative property]\label{thm:fmove-commute}
  Let $M$ be an MDS and $f_1, f_2\in\Sigma^{k-1}$ be two valid F-moves in $M$, then $f_1$ is a valid F-move in $f_2M$, $f_2$ is valid in $f_1M$, and $f_1f_2M = f_2f_1M$.
\end{lemma}
\begin{proof}
  The left companions of $f_1$ and $f_2$ are all in different PCRs.
  Hence, after doing the F-move $f_1$ or $f_2$, the other F-move is still valid.
  Moreover, regardless of the order in which the F-moves are performed, the resulting set is the same.\qed
\end{proof}

By extension, in a chain of F-moves, reordering the F-moves, as long as it is valid, does not change the result.
Note that there is no equivalent statement for I-moves: if $\im{f_{1}}{m_{1}}, \im{f_{2}}{m_{2}}$ are two valid I-moves in $M$, then $\im{f_{2}}{m_{2}}$ may not be valid in $\im{f_{1}}{m_{1}}M$.

In the following proofs, we use the simplified representation for PCRs, F- and I-moves given in Figure~\ref{fig:representation}.
For simplicity, the figure shows an example with the binary alphabet.
When $\sigma > 2$, an F-move $f$ represents a hyperedge between $\sigma$ PCRs rather than a simple edge as shown.

\mdsgraphcomponents*

\begin{proof}[Points~\ref{item:cycle-length} and~\ref{item:fmoves-cycles}, length of cycles]
  Every PCR is a cycle in $D_{k}$ and an MDS $M$ is seen as pebbles sitting on the \kmers (see Figure~\ref{fig:representation} b)
  There is one pebble per PCR.
  An F-moves involves $\sigma$ distinct PCRs (edges $af \rightarrow fa, a \in \Sigma$ are each in their own PCR).
  Hence an F-moves is an hyperedge connecting $\sigma$ PCRs.
  An F-move is like moving the pebbles along $\sigma$ PCRs at a time, from left-companions to right-companions, and this move is legal only if $\lc(f) \subset M$.
  In that sense, an F-move is like a semaphore: pebbles can move only if all their left-companions are present in the set.

  First, because every MDS has a valid F-move and a component of \gmds is finite, a component must have a cycle.
  Let $C = (M_{0}, \ldots, M_{n-1})$ be a cycle of MDSs in \gmds, and equivalently $C = (f_{0}, \ldots, f_{n-1})$ is a list of F-moves such that $M_{i+1} = f_{i}M_{i}$ (indices taken modulo $n$).
  After doing F-move $f_{0}$, the pebble on at least one PCR, say $P_{0}$, has moved.
  Because $C$ is a cycle, by the time $f_{n-1}$ is done, all pebbles are back on their respective starting spot.
  Meaning the pebble on $P_{0}$ went all the way around (possibly multiple times) $P_{0}$.
  To move around $P_{0}$ with F-moves, the pebbles in the PCR adjacent to $P_{0}$ must have moved as well, and, by the time $f_{n-1}$ is done, go around their respective PCRs.
  By transitivity, and because the de~Bruijn graph is strongly connected, every pebble on every PCR has gone around its PCR after $f_{n-1}$ is done.
  Because every node went around its PCR, this means that every one of the $\sigma^{k-1}$ F-moves was done and $n \ge \sigma^{k-1}$.

  Conversely, because the F-move/hyperedge act as semaphores, it is not possible for a pebble on a PCR to do more rotations around its own PCR than the pebbles on the adjacent (by hyperedge) PCRs.
  To see this, consider the starting position of the pebble on PCR $P_{0}$.
  For this pebble to start a second turn around $P_{0}$, all of its left-companions must be back on their starting spot and also start a second turn around their own PCRs.
   This holds for all PCRs by transitivity.

  Hence, in a cycle of the MDS graph, the pebbles of all PCRs go around the same number of times, say $\alpha$, and the number of F-moves in the cycle $C$ is $n = \alpha \sigma^{k-1}$.
  \qed
\end{proof}

\begin{figure}
  \centering
  \includegraphics{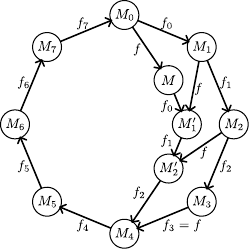}
  \caption{Example of a cycle in $\gmds(2, 4)$.
    The outer circle is $C = (f_{0}, \ldots, f_{7})$, a cycle of length $\sigma^{k-1}$.
    $M = fM_{0}$ is a neighbor of $M_{0}$ not on $C$.
    Because $f$ must occur in $C$, here $f = f_{3}$, then $f$ commutes with $f_{0}, f_{1}, f_{2}$.
    Hence $(f = f_{3}, f_{0}, f_{1}, f_{2}, f_{4}, \ldots, f_{7})$ is also a cycle in $\gmds(2, 4)$ and it contains $M_{0}$ and $M$.
  }
  \label{fig:strongly-connected}
\end{figure}

\begin{proof}[Point~\ref{item:strongly-connected}, strongly connected]
  As in the previous proof, there exists a cycle $C=(M_{0}, \ldots, M_{n-1})$ in \gmds, and its edges are $(f_{0}, \ldots, f_{n-1})$ with $M_{i+1}=f_{i}M_{i}$.

  We show that for any node $M_{i}$ of this cycle and any neighbor $M$ of $M_{i}$, reachable by an F-move or RF-move from $M_{i}$, $M$ and $M_{i}$ are in a cycle.
  If this holds, by transitivity of the relation ``being in the same strongly-connected component'', any pair of nodes in the component are in a cycle and the component is strongly-connected.

  WLOG, let's prove it for $M_{0}$ (see Figure~\ref{fig:strongly-connected}).
  It is a consequence of the commutativity of the F-moves (Lemma~\ref{thm:fmove-commute}).
  Let $M = fM_{0}$ be a neighbor of $M_{0}$ for some $f \ne f_{0}$.
  Because in a cycle all F-moves occur, there exists a first $j \in [1, n-1]$ such that $f_{j} = f$ (and $f \ne f_{i}, i \in [0, j-1]$).
  $f$ is valid in $M_{0}$, hence it is also valid in $M_{1}$, and recursively in , $M_{2}, \ldots, M_{j}$.
  Therefore $f$ commutes with $f_{0}, \ldots, f_{j-1}$ and the series of F-move $(f = f_{j}, f_{0}, \ldots, f_{j-1})$ is another path from $M_{0}$ to $M_{j+1}$ that is going through $M$.
  This path followed by the remainder of $C$ from $M_{j+1}$ back to $M_{0}$ is a cycle that includes both $M_{0}$ and $M$.
  \qed
\end{proof}

\begin{proof}[Point~\ref{item:nodes-cycles}, cycle length $\sigma^{k-1}$]
  Let $M$ be an MDS on a cycle ${C}$ in \gmds.
  It is of length $\alpha\cdot\sigma^{k-1}$, with $\alpha \ge 1$ by point~\ref{item:cycle-length}.
  Suppose that $\alpha > 1$.
  Let $C=(f_1,\ldots,f_{\alpha\cdot\sigma^{k-1}})$ be the chain of F-moves representing that cycle.
  Every distinct F-move occurs exactly $\alpha$ times in that chain.
  We show that the chain can be reordered so that the $\sigma^{k-1}$ different F-moves occur at the first  $\sigma^{k-1}$ positions of the chain.

  If it is not already the case that the first $\sigma^{k-1}$ F-moves are distinct, there must be an F-move $f$ that occurs twice in the list before an F-move $g$ occurs for the first time.
  Let $i<j$ be two indices which are the first two occurrences of $f$ in the chain (i.e., $f_i=f_j=f$), and such that $j+1$ is the first occurrence of $g$ ($f_{j+1}=g$).
  If any of the PCRs involved in the F-move $f$ are also involved in the F-move $g$, then it is not possible to use $f$ twice in $C$ before using $g$ (see Figure~\ref{fig:representation}d).
  Therefore the PCRs involved in the F-moves $f$ and $g$ are distinct, and $g$ must be a valid F-move just before the second use of $f$ as well.
  In other words, $f_j$ and $f_{j+1}$ commute.

  Repeated swapping of F-moves leads to the desired chain of F-moves with all $\sigma^{k-1}$ distinct F-moves in the first positions, which induces a cycle of length $\sigma^{k-1}$ containing $M$.
  \qed
\end{proof}

\begin{proof}[Point~\ref{item:thickle}, $\sigma^{k-1}$-partite]
  Partition the nodes of a component of \gmds as follows.
  We create $\sigma^{k-1}$ sets: $\mathcal{P}_{0}, \ldots, \mathcal{P}_{\sigma^{k-1}-1}$.
  Let $M_{0}$ be an arbitrary MDS of the component and assign it to the set $\mathcal{P}_{0}$.
  For every other MDS $M$, take a shortest path $P(M) = M_{0} \rightarrow M$ in \gmds.
  Assign $M$ to the partition with index $|P| \mod \sigma^{k-1}$.

  Because $M_{0}$ is in a cycle of length $\sigma^{k-1}$, every set $\mathcal{P}_{i}$ has at least one MDS assigned to it.
  Moreover, every MDS is assigned to exactly one set.
  Hence the sets $\mathcal{P}_{i}$ form a partition of the MDSs in the component.

  An edge between MDSs in sets $\mathcal{P}_{i}$ and $\mathcal{P}_{j}$ with $j > i+1$ would imply the existence of a cycle containing $M_{0}$ of length $< \sigma^{k-1}$, which is not possible. \qed
\end{proof}

\section{Cycle signature unique per component}\label{sec:cycle-sign-uniq}

An MDS $M$ is called \emph{$f$-terminal} if the only valid F-move in $M$ is $f$.

\begin{lemma}\label{thm:exist-f-terminal}
  For any $f \in \Sigma^{k-1}$ and in any component of \gmds, there exists an $f$-terminal MDS.
\end{lemma}
\begin{proof}
  From Proposition~\ref{thm:comp-structure}, in any component there exists an MDS $M'$ where $f$ is a valid F-move.
  If there exists other valid F-moves than $f$ in $M'$, do them recursively.
  I.e., we do every possible F-move in $M'$ but refuse to do $f$.
  This creates a path $P$ of MDSs in \gmds starting at $M'$ that does not contain $f$ as an edge.

  Because every cycle in \gmds contains every possible F-move, $P$ cannot induce a cycle, and it must terminate at an MDS $M$.
  By construction $M$ is $f$-terminal.
  \qed
\end{proof}

An $f$-terminal MDS $M$ has a useful property: every maximal path in $D_{k}$ that avoids $M$ (as created by a walk like in Proposition~\ref{thm:f-moves-exist}) must start at a \kmer $m \in \rc(f)$.
Equivalently, any walk in $D_{k}$ that avoids $M$ following edges backward ends at some $m \in \rc(f)$.

% TODO: does the following proof only work for binary alphabet? If not, explain the 0.

\signatures*
\begin{proof}[Point~\ref{item:sig-diff}, different signatures]
  Fix $f \in \Sigma^{k-1}$ and by Lemma~\ref{thm:exist-f-terminal} we can assume that $M_{1}$ and $M_{2}$ are both $f$-terminal, each in its own component.
  We will construct a cycle $C$ in $D_{k}$ that has different hitting numbers between the components: $\hit_{M_{1}}(C) \ne \hit_{M_{2}}(C)$.

  $M_{1}$ and $M_{2}$ are in different components, so they are distinct MDSs and there exists a PCR $R$ where the selected \kmer is different.
  That is, $R \cap M_{1} \triangleq m_{1} \ne m_{2} \triangleq R \cap M_{2}$.
  Take a path $P$ in $D_{k}$ following edges backward from node $0f$ (which is in both $M_{1}$ and $M_{2}$) to $m_{1}$ that avoids nodes $af, a \in \Sigma \setminus \{0\}$.
  Path $P$ exists because $D_{k}$ is $(\sigma-1)$-connected.
  Because $m_{1} \in M_{1} \SymDiff M_{2}$, there must exist a first node $m \in P$ which is in $M_{1} \SymDiff M_{2}$.

  Let $P_{1}$ be the restriction of the path $P$ from $0f$ to $m$ and, WLOG, assume that $m \in M_{1}$.
  By construction, $|P_{1} \cap M_{1}| = |P_{1} \cap M_{2}| + 1$.

  Let $P_{2}$ be a path created by a maximal random walk in $D_{k}$, following edges backward, starting from $m$ and that avoids $M_{2}$.
  Because $M_{2}$ is $f$-terminal, the walk ends at a node $fa \in \rc(f), a \in \Sigma$.
  By construction, $|P_{2} \cap M_{1}| \ge |P_{2} \cap M_{2}| = 0$ ($P_{2}$ avoids nodes from $M_{2}$ but may contain nodes from $M_{1}$).

  Two cases can happen.
  First case, there exists a first node $m' \in P_{1} \cap P_{2}$.
  Then define the cycle $C$ as the restriction of $P_{1}$ from $m'$ to $m$ followed by the restriction of $P_{2}$ from $m$ to $m'$.
  Second case, $P_{1} \cap P_{2} = \emptyset$ and define the cycle $C$ as the concatenation of $P_{1}, P_{2}$ and backward edge $fa \rightarrow 0f$.

  In both cases, $C$ satisfies by construction $\hit_{M_{1}}(C) > \hit_{M_{2}}(C)$.
  \qed
  \end{proof}

\section{\gcomp is undirected}\label{sec:gcomp-undirected}

\begin{figure}
  \centering
  \includegraphics[scale=1.2]{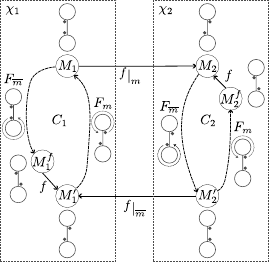}
  \caption{Simplified example for finding the complementary I-moves, when $\sigma = 2$.
    On the left box, component $\chi_{1}$ and component $\chi_{2}$ on the right, of \gmds.
    The cycle $C_{1}, C_{2}$ are cycles in $\chi_{1}$ and $\chi_{2}$ respectively.
    The simplified PCR/pebble drawings represent the position of the pebbles on the PCRs of $P_{m}$ (top PCR) and $P_{\overline{m}}$ (bottom PCR).
    The edge between these PCRs represents $f$.
    The PCR/pebbles drawings next to the MDS nodes represent the state of the PCRs for these MDSs, while the drawings next to the F-move lists represent the action of the list of F-moves on the pebbles.
    From the cycle $C_{1}$ in $\chi_{1}$, we construct cycle $C_{2}$ in $\chi_{2}$ by swapping the order of the F-moves: $(F_{\overline{m}}, f, F_{m}) \rightarrow (F_{m}, f, F_{\overline{m}})$.
    These cycles go through the desired MDSs $M'_{2}$ and $M'_{1}$ that are linked by the complementary I-move $\im{f}{\overline{m}}$.
  }
  \label{fig:complement-imove}
\end{figure}

\complementimove*
\begin{proof}
  See Figure~\ref{fig:complement-imove}.
  In component $\chi_{1}$, by Proposition~\ref{thm:comp-structure}, there is a cycle $C_{1}$ of length $\sigma^{k-1}$ that contains MDS $M_{1}$, and this cycle has $f$ has an F-move.
  Hence, $C_{1} = (M_{1}, \ldots, M^{f}_{1}, M'_{1}, \ldots)$ where $M^{f}_{1}$ is the MDS where $f$ is a valid I-move and $M'_{1} = fM^{f}_{1}$.
  Equivalently, looking at the edges, $C_{1} = (F_{\overline{m}}, f, F_{m})$ where $F_{\overline{m}}, F_{m}$ are lists of F-moves.

  In $M_{1}$, $\im{f}{m}$ is a valid I-move, which means that if $m_{a} = 1$, then $af \in M_{1}$ and $fa \in M_{1}$ otherwise.

  Let's call $P_{m}$ the set of PCRs that contain $af$ when $m_{a} = 1$, and $P_{\overline{m}}$ the PCRs containing $af$ when $m_{a} = 0$ ($P_{m}$ contains only the top PCR in Figure~\ref{fig:complement-imove}, and $P_{\overline{m}}$ the bottom PCR).

  In $M^{f}_{1}$, $f$ is a valid F-move, which means that $af \in M^{f}_{1}$ for all $a \in \Sigma$.
  In other words, the list of F-moves $F_{\overline{m}}$ made by the pebbles in the PCRs in $P_{\overline{m}}$ go around from $fa$ to $af$, while the pebbles in the PCRs in $P_{m}$ did not move.
  (The only way for the pebbles in the PCRs in $P_{m}$ to move is to do F-move $f$, which by construction is not in $F_{\overline{m}}$).

  Similarly, the list of F-moves $F_{m}$ made by the pebbles in the PCRs in $P_{m}$ go around from $fa$ to $af$, while the pebbles in the PCRs of $P_{\overline{m}}$ did not move.

  Now from $M_{1}$ do the valid I-move $\im{f}{m}$.
  This advances the pebbles in the PCRs of $P_{m}$ from $af$ to $fa$ (forward by 1 edge), to get to $M_{2}$ in component $\chi_{2}$, where $\rc(f) \subset M_{2}$.
  The position of the pebbles in $M_{1}$ and $M_{2}$ agree everywhere except on the PCRs of $P_{m}$.
  Because the F-moves in $F_{\overline{m}}$ do not affect the PCRs of $P_{m}$, the list $F_{\overline{m}}$ is a valid list of F-moves in $M_{2}$ as well.

  $fa \in M_{2}$ for all $a \in \Sigma$.
  Applying $F_{\overline{m}}$ to $M_{2}$ leads to MDS $M'_{2}$ where $af \in M'_{2}$ if $m_{a} = 0$ and $fa \in M'_{2}$ otherwise.
  In other words, I-move $\im{f}{\overline{m}}$ is valid in $M'_{2}$.
  It is easy to check that doing the I-move $\im{f}{\overline{m}}$ gets back to $M'_{1}$.

  For completion, one can check that the list of F-moves $F_{m}$ applies to $M'_{2}$ because $M'_{2}$ and $M'_{1}$ only differs on the pebbles on the PCRs of $P_{\overline{m}}$ and $F_{m}$ does not affect those PCRs.
  Applying $F_{m}$ get to $M^{f}_{2}$ where $f$ is a valid F-move and $M_{2} = fM^{f}_{2}$.

  Therefore, the cycle $C_{1} = (F_{\overline{m}}, f, F_{m})$ is a valid cycle in $\chi_{1}$ and contains $M_{1}$ and $M'_{1}$, while $C_{2} = (F_{m}, f, F_{\overline{m}})$ is valid in $\chi_{2}$ and contains $M_{2}$ and $M'_{2}$.
  \qed
\end{proof}

\section{Non-decycling PCR sets}\label{sec:non-decycling-pcr}

\begin{proposition}
  Let \gpcr\ be the graph with non-decycling PCR sets as nodes and F-moves as edges.
  Then each component of $G$ is a DAG.
\end{proposition}
\begin{proof}
  Suppose there exists a cycle $\mathcal{C} = \{ M_{1}, \ldots, M_{n} \}$ in \gpcr, where $M_{i+1} = f_{i} M_{i}$.
  Because $M_{1}$ is not decycling, then there exists a cycle $C$ in $D_{k} \setminus M_{1}$.
  Because RF-moves preserve the hitting number, $C$ is also a cycle in $D_{k} \setminus f_{1}M_{1}$, and by induction a cycle in $D_{k} \setminus M_{i}, i \in [1, n]$.
  From the proof Proposition~\ref{thm:comp-structure}, any cycle $\mathcal{C}$ must do every $\sigma^{k-1}$ F-move to return to the starting set, and the union of all the left-companions of the F-moves is the set of all \kmers.
  This is a contradiction.
\end{proof}

\section{I-move and constrained cycles}\label{sec:i-move-constrained}

\begin{figure}
  \centering
  \includegraphics[scale=1.2]{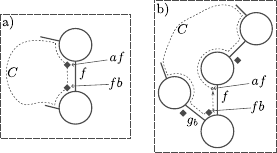}
  \caption{a)~The $\im{f}{m}$ with $m_a = 1$ and $m_{b} = 0$ is not possible because $\hit_{M}(C) = 1$.
    When the I-move $\im{f}{m}$ is valid, necessarily $C$'s hitting number must be at least $2$.
    b)~Suppose $\im{f}{m}$ is never valid, then a backward walk creates a cycle with hitting number $1$ using the edge $af \rightarrow fb$.
  }
  \label{fig:constrained-cycle}
\end{figure}

\imoveconstrained*
\begin{proof}
  Let $\im{f}{m}$ be a potential I-move with $m_{a} = 1$ and $m_{b} = 0$ ($a, b \in \Sigma, a \ne b$).

  Suppose there exists a constrained cycle $C$ in the de~Bruijn graph $D_{k}$ using the edge $af \rightarrow fb$, and $H_{\chi}(C) = 1$.
  If $\im{f}{m}$ is a valid I-move in an MDS $M \in \chi$, then by definition $af, fb \in M$, hence $H_{M}(C) \ge 2$.
  This contradict that $C$ is constrained (see Figure~\ref{fig:constrained-cycle} a).

  Conversely, suppose that $\im{f}{m}$ is not a valid I-move in any MDS of $\chi$.
  Let $M^{f} \in \chi$ be an MDS where $f$ is a valid F-move and $M = fM^{f}$.
  Then $\rc(f) \subset M$.
  Define $g_{c} \triangleq f[2:k-2]c, c\in\Sigma$, that is for all right-companion of $f$, $fc \in \lc(g_{c})$.

  From $M$ recursively do all valid F-moves except for the F-moves $g_{c}$ where $m_{c} = 0$ to obtain $M' \in \chi$ where the only valid F-moves are exactly those than we refused to do.
  There must exist $a \in \Sigma$ such that $m_{a} = 1$ and $af \notin M'$, otherwise $\im{f}{m}$ is a valid I-move in $M'$ (see Figure~\ref{fig:constrained-cycle}b).
  From $af$ do a walk that avoids $M'$ using backward edges.
  This walk must end at one of the right-companions of the valid F-moves in $M'$, that is there exists $b$ such that walk ends at $m' \in \rc(g_{b})$.
  By construction there is a backward edge $m' \rightarrow fb$.
  Then follow the backward edge $fb \rightarrow af$ to create a cycle $C$.
  By construction the only node from $M'$ in cycle $C$ is $fb$, hence $\hit_{M'}(C) = 1$ and $C$ uses the edge $af \rightarrow fb$ with $m_{a} = 1$ and $m_{b} = 0$.
  \qed
\end{proof}

\end{document}